\documentclass[12pt,onecolumn,journal]{IEEEtran}
\usepackage{graphicx,amssymb,amsmath,amsthm,cite,cases}

\def\ds{\displaystyle}

\def\Join{\hat{\operatorname{J}}}

\def\be{\begin{equation}}
\def\ee{\end{equation}}
\def\ben{\begin{eqnarray}}
\def\een{\end{eqnarray}}
\def\Sti{st_{\rm{ind}}}
\def\Stc{st^j_{\rm{cf}}}
\def\Sts{st^j_{\rm{sg}}}

\def\D{\mathcal{D}}
\def\R{\mathbb{R}}
\def\N{\mathbb{N}}
\def\V{\mathbb{V}}

\def\Re{\operatorname{Re}}
\def\Im{\operatorname{Im}}
\def\im{\operatorname{\imath}}

\def\qs{q^\star}

\def\jq{j^\star_{q}}

\def\til{\tilde}

\def\vR{\mathbf{R}}
\def\vd{\mathbf{d}}
\def\vs{\mathbf{s}}
\def\vst{\mathbf{\til{s}}}
\def\st{\til{s}}
\def\ellt{\til{\ell}}

\def\vc{\mathbf{c}}
\def\vct{\mathbf{\til{c}}}

\def\ctq{\til{c}^\Delta}
\def\ct{\til{c}}
\def\vf{\mathbf{f}}
\def\vfr{\mathbf{f}^{\rm{r}}}
\def\vfk{\mathbf{f}^{\kq}}
\def\vF{\mathbf{F}}
\def\vFr{\mathbf{F}^{\rm{r}}}
\def\vg{\mathbf{g}}
\def\vh{\mathbf{h}}
\def\vy{\mathbf{y}}

\def\vB{\mathbf{b}}

\def\vW{\mathbf{w}}
\def\vWt{\til{\mathbf{w}}}

\def\vR{\mathbf{R}}
\def\vr{\mathbf{r}}
\def\vd{\mathbf{d}}

\def\op{\hat{P}}
\def\kq{k_q}
\def\jq{q,j}

\def\Nq{N_b}
\def\jL{j=1,\ldots,L}
\def\qQ{q=1,\ldots,Q}

\newcommand{\la}{\langle}
\newcommand{\ra}{\rangle}

\newtheorem{corollary}{Corollary}
\newtheorem{definition}{Definition}
\newtheorem{theorem}{Theorem}

\def\qQ{q=1,\ldots,Q}
\def\jL{j=1,\ldots,L}

\newcommand{\Spann}{{\mbox{\rm{span}}}}

\title{Trigonometric dictionary based codec 
for music compression with high quality recovery}

\author{Laura Rebollo-Neira\\
Mathematics Department\\
Aston University\\
B3 7ET, Birmingham, UK}

\begin{document}
\maketitle 
\baselineskip = 2\baselineskip
\begin{abstract}
A codec for compression of music signals is proposed.
The method belongs to the class of transform 
lossy compression. It is conceived 
to be applied in the high quality recovery range though.
The transformation, endowing  the codec with 
its distinctive feature, 
relies on the ability to construct high quality
 sparse approximation of music
signals. This is achieved 
by a redundant trigonometric
dictionary and a dedicated pursuit strategy. 
The potential of the approach is illustrated 
by comparison with the OGG  Vorbis format, on a 
sample consisting of clips of melodic music. The 
comparison evidences remarkable improvements in 
compression performance for the identical 
quality of the decompressed signal.
\end{abstract}
Trigonometric Dictionaries.
\section{Introduction}
For the most part the techniques for compressing 
high fidelity music have been developed within the lossless 
compression framework
\cite{HS01,LR04,YRL06,HFH08,GT08,NZ11,ABG12,GT13}.
Because lossless music compression algorithms 
are reversible, 
which implies that can reproduce the original signal  
when decompressing the file, the efficiency of 
those algorithms are compared on the reduction of file size
and speed or the process.
Conversely, lossy compression introduces irreversible
loss and should be compared also taking into 
account the quality of the decompressed data. 

This work focusses on lossy compression of music signals
with {\em{high quality recovery}}. This means 
that the recovered signal should be very similar to 
the original one, with respect to the 
Euclidean distance of the data points. 
In other words, the recovered 
signal should yield a high Signal to Noise 
Ratio (SNR). 
The proposed approach is based on the ability to  
construct a high quality sparse representation of  
a piece of music. The sparsity is 
achieved by selecting elements from a redundant 
trigonometric dictionary, through a 
 dedicated greedy pursuit methodology which 
approximates simultaneously all the channels of a 
stereo signal. Pursuit strategies for  
approximating multiple signals sharing the same
sparsity structure are refereed in the  
literature to as  several names: Vector greedy algorithms
\cite{LT03,LT06}, simultaneous greedy approximations 
\cite{LT06, TGS06}, and multiple measurement vectors (MMV). 
\cite{CRE05,CH06}. 
Following previous work \cite{RNMB13,LRN15}, 
we dedicate greedy methodologies for simultaneous 
representation to approximate a 
{\em{partitioned multichannel 
music signal}} subjected to a global constraint on 
the sparsity.

The use of redundant dictionaries for 
constructing sparse representations is known to
 be a successful approach in a variety 
of signal processing 
applications 
\cite{PAB06,FVP06,MES08,WYG09,Ela10,WMM10,YUW10,PBD10}.
 In particular for compression of facial images
\cite{BE08,ZGK11,RCE14}.
This paper extends the range of 
successful applications 
by presenting a number of examples were the proposed 
dictionary based codec 
for compression of music signals 
yields remarkable results, in relation to file size
and quality of the recovered signal.
\subsection{Paper contributions}
The central aim of the paper is to 
produce a proof of concept of the proposed 
codec.
The proposal falls within the usual
transform coding scheme. It consists of three
main steps:
\begin{itemize}
\item
[i)] Transformation of the signal.
\item
[ii)] Quantization of the transformed data.
\item
[iii)] Bit-stream entropy coding.
\end{itemize}
However, we move away from the traditional 
compression techniques at the very beginning.
Instead of considering an orthogonal transformation,
the first step is realized by approximating 
the signal using a trigonometric redundant dictionary. 
In a previous work \cite{LRN15}, the dictionary
has been proven to yield  stunning 
 sparse approximation of melodic music, 
if processed by the adequate greedy
strategy. We demonstrate now that the sparsity  
 renders compression.

The Hierarchized Block Wise
Optimized Orthogonal Matching Pursuit (HBW-OOMP) method
in \cite{LRN15} is generalized here, to consider 
the simultaneous approximation
of multichannel signals. Within the proposed 
scheme the advantage of simultaneous 
approximation is twofold: a)It reduces the processing 
time at the transformation stage and b)It reduces 
the number of parameters to be stored, which improves 
 compression performance.

The success of the codec, designed to 
achieve high quality recovery, is illustrated
by comparison with the OGG Vorbis compression format. 
Accessing the quality of the recovered signal by the 
classic SNR, a substantial gain in 
compression, for the same quality of 
the decompressed signal,
is demonstrated on a number of clips 
of melodic music.
\subsection{Paper Organization}
Sec.~\ref{notation} introduces the notation and some 
relevant mathematical background.
Sec.~\ref{MHBW} discusses the HBW 
strategy to approximate simultaneously a 
 multichannel signal.
Sec.~\ref{cod} describes a simple compression scheme
 that benefits from the achieved sparse approximation 
of the multichannel signal.
 Sec.~\ref{NE} demonstrates the potential of the
technique by comparison with the OGG format. 
The final conclusions are presented in
 Sec.~\ref{con}
\section{Mathematical background and notational}
\label{notation}
Throughout the paper $\R$ and $\N$ stand for 
the sets of real and natural numbers, respectively.
Low boldface letters are used to indicate Euclidean 
vectors and capital boldface letters to indicate matrices. 
Their corresponding component are represented using 
standard  mathematical fonts,
e.g., $\vf \in \R^N,\, N \in \N$ is a vector of components
$f(i),\, i=1,\ldots,N$ and $\vF \in \R^{N \times L}$ is a
matrix of real entries 
$F(i,j),\,i=1,\ldots,N,\, j=1,\ldots,L$.

An $L$-channel signal is represented as a matrix 
$\vF \in \R^{N \times L}$
the columns of which are the channels, indicated as vectors
$\vf_j \in \R^N,\,j=1,\ldots,L$. Thus, a single 
channel reduces to a vector.
A partition of a multichannel signal 
$\vF \in \R^{N \times L}$ is realized by  
 a set of disjoint pieces $\vF_q \in \R^{\Nq \times L},\, 
q=1,\ldots,Q$, which for simplicity are assumed to
be all of the same size and such that $Q \Nq=N$,
i.e., for each channel it holds that 
$\vf_j = \Join_{q=1}^{Q} \vf_{\jq}$, where the 
concatenation operation $\Join$ is 
defined as follows: $\vf_j$ is a vector  
in $\R^{Q \Nq}$ having components 
$f_j(i)=f_{\jq}(i-(q-1)\Nq),\, i=(q-1)\Nq+1,\ldots,q\Nq,\,q=1,\ldots,Q$. In the adopted notation $f_j(i)$ can also be 
indicated as $F(i,j)$ and 
$f_{\jq}(i)$ as $F_q(i,j)$.
Hence $$\|\vF\|_F^2 =  \sum_{q=1}^Q\|\vF_q\|_F^2,$$ 
where each 
$\vF_q$ is a matrix consisting of the channels 
$\vf_{\jq} \in \R^{\Nq},\,j=1,\ldots,L$, as columns 
and $\| \cdot \|_F$ indicates the Frobenius norm. 
Accordingly,
  $$\| \vF_q\|_F^2 = \sum_{j=1}^L 
\| \vf_{\jq}\|^2,$$
where $\| .\|$ indicates the Euclidean norm induced 
by the Euclidean inner product $\la \cdot, \cdot \ra$.
\begin{definition} 
A {\em{dictionary}} for $\R^{\Nq}$ is  
 an {\em{over-complete}} set of normalized to unity 
elements 
$\D=\{\vd_n \in \R^{\Nq}\,; \| \vd_n\|=1\}_{n=1}^M,$ 
 which are called {\em{atoms}}.
\end{definition}
{\bf{Approximation assumption:}}
{\em{Given a dictionary $\D$ and a multichannel signal 
partitioned into $Q$ blocks $\vf_{\jq} \in \R^{\Nq},\,\jL,\qQ$,
as described above, the $\kq$-term approximation 
of each block is assumed to be  of the form 
\be
\label{atoq}
\vfk_{\jq}= \sum_{n=1}^{\kq}
c_{\jq}^{k_q}(n) \vd_{\ell^{q}_n},
\quad \jL,\,\qQ,
\ee
where the atoms $\vd_{\ell^{q}_n},\,n=1,\ldots,\kq$ 
are {\em{the same}} for all the channels corresponding 
to a particular block $q$.}}

Before discussing how to select the atoms in \eqref{atoq}
it is convenient to review some 
properties of an orthogonal projector. Let's start
by recalling its definition.
\begin{definition}
An operator $\op_{\V_{\kq}^q}$ is an orthogonal projection
operation onto ${\V_{\kq}^q} \subset \R^{\Nq}$ if and only if:
\begin{itemize}
\item[a)]
$\op_{\V_{\kq^q}}$ is idempotent,
i.e.,
$\op_{\V_{\kq^q}}\op_{\V_{\kq^q}} = \op_{\V_{\kq^q}}$.
\item[b)]
$\op_{\V_{\kq^q}} \vg =\vg$ if $\vg \in {\V_{\kq^q}}$ and
$\op_{\V_{\kq^q}} \vg^\perp=0$ if $\vg^\perp
\in {\V_{\kq^q}^\perp}$,
with ${\V_{\kq^q}^\perp}$
indicating the orthogonal complement
of $\V_{\kq^q}$ in $\R^{\Nq}$.
\end{itemize}
\end{definition}
The following properties will be used in the
proofs of subsequente theorems.
\begin{itemize}
\item[i)]
An orthogonal projector operator
 $\op_{\V_{\kq^q}}$ is hermitian, i.e. for all
$\vh$ and $\vg$ in $\R^{\Nq}$ it is true that
$\la \vh, \op_{\V_{\kq}^q} \vg\ra=
\la \op_{\V_{\kq}^q} \vh,  \vg\ra.$
\item[ii)]
If ${\V_{\kq+1}^q}$ is constructed as
${\V_{\kq+1}^q}= \V_{\kq}^q + \vd_{\ell_{\kq+1}^q}$,
for all $\vh \in \R^{\Nq}$ it holds that
$$\op_{\V_{\kq+1}^q} \vh= \op_{\V_{\kq}^q} \vh + 
\vW^q_{{\kq+1}} \frac{\la \vW^q_{{\kq+1}},\vh\ra}
{\|\vW^q_{{\kq+1}}\|^2},$$ with $\vW^q_{{\kq+1}}= 
\vd_{\ell_{\kq+1}^q} - \op_{\V_{\kq}^q} \vd_{\ell_{\kq+1}^q}.$
\end{itemize}
\begin{theorem}
\label{the1}
Let $\vF_q$ be a $\Nq \times L$ matrix the columns of which are the $L$ signals $\vf_{\jq} \in \R^\Nq, \, j=1,\ldots,L$ and let
$\vF_q^{\kq}$ be the matrix with the corresponding
$\kq$-term
approximations
$\vf_{\jq}^{\kq}\in  \V_{\kq^q}, \,j=1,\ldots,L$.
For the error $\|\vF_q - \vF_q^{\kq}\|_F^2$ to be
minimum the $\kq$-term
approximations must satisfy:
$\vf_{\jq}^{\kq} = \op_{\V_{\kq^q}} \vf_{\jq},\, 
j=1,\ldots,L$.
\end{theorem}
\begin{proof}
$\|\vF_q - \vF_q^{\kq}\|_F^2$ can be expressed as
\ben
\|\vF_q - \vF_q^{\kq}\|_F^2\! &=& \!
\sum_{j=1}^L
\la \vf_{\jq}- \vf_{\jq}^{\kq},\vf_{\jq}- \vf_{\jq}^{\kq}\ra\nonumber\\
 \!&=& \!\sum_{j=1}^L \| \vf_{\jq}\|^2 -2\la \vf_{\jq}, \vf_{\jq}^{\kq}\ra + \|\vf_{\jq}^{\kq}\|^2.
\label{ernorm}
\een
Since $\vf_{\jq}^{\kq}$ is an element of ${\V_{\kq}^q}$
we can write it as $\vf_{\jq}^{\kq} =  
\op_{\V_{\kq}^q} \vf_{\jq} + \vg_{\jq}$ for some $\vg_{\jq} \in \V_{\kq}^q$.  The corresponding
replacements in \eqref{ernorm}, and the fact that
$\vg_{\jq} = \op_{\V_{\kq}^q} \vg_{\jq}$ and $\op_{\V_{\kq}^q}$
is idempotent and hermitian, lead to the expression
\ben
\|\vF_q - \vF_q^{\kq}\|_F^2=
\sum_{j=1}^L \| \vf_{\jq}\|^2 - \la \vf_{\jq},\op_{\V_{\kq}} \vf_{\jq} \ra + \| \vg_{\jq} \|^2, 
\een
from where it follows that $\|\vF_q - \vF_q^{\kq}\|_F^2$
is minimum if $\vg_{\jq}=0, \, j=1,\ldots,L$ 
i.e. $\vf_{\jq}^{\kq}=
\op_{\V_{\kq}} \vf_{\jq},\, j=1,\ldots,L$.
\end{proof}
\begin{corollary}
The statement of
Theorem \ref{the1} also minimizes the
norm of the total residual $\vR^{K}= \vF - \vF^K$ 
in approximating the whole multichannel 
signal, with $\vF^K= \Join_{q=1}^Q \vF_q^{\kq}$ and
$K=\sum_{q=1}^Q \kq$. 
\end{corollary}
\begin{proof}
It readily follows by the a
definition of the adopted disjoint partition: 
$$\|\vR^K\|_F^2= \sum_{q=1}^Q 
\|\vF_q - \vF_q^{\kq}\|_F^2$$ is obviously minimum if 
each $\|\vF_q - \vF_q^{\kq}\|_F^2$ is 
minimum.
\end{proof}
Assuming, for the moment, that the sets of indices 
$\Gamma_q=\{\ell^q_n\}_{n=1}^{\kq}$ labeling 
the atoms in \eqref{atoq} are known,
we recall at this point an effective construction  of 
the required orthogonal projector for optimizing the
approximation. Such a projection is given in terms of
biorthogonal vectors as follows: 
\be
\label{oproj}
\hat{P}_{\V_{k}^q} \vf_{\jq} = \sum_{n=1}^{\kq}
 \vd_{\ell^{q}_n} \la \vB_n^{\kq,q}, \vf_{\jq} \ra =
\sum_{n=1}^{\kq} c_{\jq}^{\kq}(n)\vd_{\ell^{q}_n}.
\ee
For a fixed $q$
the vectors $\vB_{n}^{\kq,q},\,n=1,\ldots,k_q$ are biorthogonal to the selected atoms
$\vd_{\ell_n^{q}},\, n=1,\ldots,k_q$ and span the identical subspace, i.e.,
$$\V_{k_q}^q= \Spann\{\vB_{n}^{k_q,q}\}_{n=1}^{k_q}=
 \Spann\{\vd_{\ell_n}^q\}_{n=1}^{k_q}.$$
Such vectors  can be adaptively constructed, from 
$\vB_1^{1,q}=\vW_1^q =\vd_{\ell_1}^q$, through the
recursion formula \cite{RNL02}:
\be
\begin{split}
\label{BW}
\vB_{n}^{k_q+1,q}&= \vB_{n}^{k_q,q} - \vB_{k_q+1}^{k_q+1,q}\la \vd_{\ell_{k_q+1}}^q, \vB_{n}^{k_q,q}\ra,\quad n=1,\ldots,\kq,\\
\vB_{k_q+1}^{k_q+1,q}&= \vW_{k_q+1}^q/\| \vW_{k_q+1}^q\|^2,
\end{split}
\ee
with 
\be
\begin{split}
\label{GS}
\vW_{k_q+1}^q= \vd_{\ell_{k_q+1}}^q - \sum_{n=1}^{k_q} \frac{\vW_n^q}
{\|\vW_n^q\|^2} \la \vW_n^q, \vd_{\ell_{k_q+1}}^q\ra.
\end{split}
\ee
For numerical accuracy in the construction
of the orthogonal set
$\vW_{n}^q,\,n=1,\ldots,k_q+1$ at
least one re-orthogonalization step
is usually needed. This
implies to recalculate the vectors as
\be
\label{RGS}
\vW_{k_q+1}^q \leftarrow \vW_{k_q+1}^q- \sum_{n=1}^{\kq} \frac{\vW_{n}^q}{\|\vW_n^q\|^2}
\la \vW_{n}^q , \vW_{k_q+1}^q\ra.
\ee

The alternative representation of $\op_{\V_{k_q}^q}$,
in terms of vectors $\vW_n^{q},\,n=1,\ldots,\kq$,
gives the decompositions:
\be
\label{oproj2}
\hat{P}_{\V_{\kq}^q} \vf_{\jq} = \sum_{n=1}^{\kq}
\vW_n^q\frac{\la \vW_n^{q},\vf_{\jq}\ra}{\|\vW_n^{q}\|^2},
\quad \jL,\,\qQ.
\ee
While these decompositions are not the
representations of interest (c.f. \eqref{atoq}) they
play a central role in the derivations of the next
section.
\section{Multichannel HBW strategy}
\label{MHBW}
The HBW version of pursuit strategies 
\cite{LRN15, RNMB13} is a dedicated implementation of
those techniques, specially designed for approximating
a signal partition subjected to a global
constraint on sparsity. The approach operates 
 by raking the partition units for their sequential stepwise
approximation. In this section we extend the
HBW  method to consider the case of a multichannel signal, 
within the approximation assumption specified in 
Sec.~\ref{notation}.
\begin{theorem}
\label{the2}
Considerer that, for each block $q$, the $\kq$-term atomic
decompositions \eqref{atoq}
fulfilling that
$\vf_{\jq}^{\kq}= \hat{P}_{\V_{\kq}^q} \vf_{\jq}$ are known,
with ${\V_{\kq}^q}=
\Spann\{\vd_{\ell^{q}_n}\}_{n=1}^{\kq}$.
Let the indices $\ell^{q}_{\kq+1}\notin 
\{\ell^{q}_n\}_{n=1}^{\kq}$ be selected,
for each $q$-value,
by the same criterion as the one used to choose the 
atoms in \eqref{atoq}.
In order to minimize the square norm of the total residual 
$\vR^{K+1}$, with $K=\sum_{q=1}^Q \kq$,
the atomic decomposition to be upgraded at iteration $K+1$
should correspond to the block $\qs$ such that
\be
\label{hbwoompml}
q^\star=
\operatorname*{arg\,max}_{q=1,\ldots,Q}
 \frac{\sum_{j=1}^L|\la \vW_{\kq+1}^q, \vf_{\jq} 
\ra|^2}{\|\vW_{\kq+1}^q\|^2},
\ee
with $\vW_{1}^q= \vd_{\ell_1^q}$ and 
$\vW_{\kq+1}^q= \vd_{\ell_{\kq+1}^q} - \op_{\V_{\kq}^q}
\vd_{\ell_{\kq+1}^q}$.
\end{theorem}
\begin{proof}
Since at iteration $K+1$ the atomic decomposition of only 
{\em{one}} block is upgraded by one atom, 
the total residue at that iteration is constructed as
$$\vR^{K+1}=\Join_{\substack{p=1\\ p \ne q}}^Q\,\vR_p^{k_p} 
\; \Join \; \vR_q^{\kq+1}.$$
Then, 
$$\|\vR^{K+1} \|_F^2= \sum_{\substack{p=1\\ p \ne q}}^Q
\|\vR_p^{k_p}\|_F^2 + \|\vR_q^{\kq+1}\|_F^2,$$
so that   $\|\vR^{K+1} \|$ is minimized by 
the minimum value of $\|\vR_q^{\kq+1}\|_F^2$.
Moreover, by definition 
$\|\vR_q^{\kq+1}\|_F^2= \|\vF_q - \vF_q^{\kq+1}\|_F^2$
and, from \eqref{ernorm} and the fact that 
$\vf_{\jq}^{\kq+1}= \op_{\V_{\kq+1}^q} \vf_{\jq}$, 
we can write: 
\ben
\label{errk1}
\|\vR_q^{\kq+1}\|_F^2
&=&\sum_{j=1}^L \| \vf_{\jq}\|^2-\| \op_{\V_{\kq+1}^q} \vf_{\jq} \|^2.
\een
Then, $\|\vR_q^{\kq+1}\|_F^2$
is minimum if 
$\sum_{j=1}^L \|\op_{\V_{\kq+1}^q}\vf_{\jq}\|^2$ is maximum.  
Applying the property ii) of an orthogonal projector 
listed in Sec.~\ref{notation} we have:
$$ \sum_{j=1}^L \|\op_{\V_{\kq+1}^q} \vf_{\jq} \|^2= 
\sum_{j=1}^L \|\op_{\V_{\kq}^q} \vf_{\jq}\|^2  + 
 \sum_{j=1}^L \frac{|\la \vW^q_{{\kq+1}},\vf_{\jq}\ra|^2}
{\|\vW^q_{{\kq+1}}\|^2},$$ with
 $\vW^q_{{\kq+1}}= 
\vd_{\ell^q_{\kq+1}} - \op_{\V_{\kq}^q} \vd_{\ell_{\kq+1}^q}.$
Because $\op_{\V_{\kq}^q} \vf_{\jq}$ is 
fixed at iteration $K+1$, we are in a position 
to conclude that  
 $\|\vR^{K+1} \|_F^2$ is 
minimized by  upgrading the 
atomic decomposition of the block $\qs$ satisfying 
\eqref{hbwoompml}.
\end{proof}
Theorem \eqref{the2} gives the HBW 
prescription optimizing the
raking of the blocks in a multichannel signal partition, 
for their sequential stepwise approximation. This is 
 irrespective of what the criterion for choosing the 
atoms for the approximation is. 
The  method as a whole depends 
on that criterion, of course. One can 
use for example an extension 
of the Orthogonal Matching Pursuit 
 (OMP) criterion, 
which when applied to multichannel 
signals has been termed 
 simultaneous OMP in \cite{TGS06}. 
 According to this criterion the indices of 
the atoms in the approximation of 
each block-$q$ are such that
\be
\label{ompml1}
\ell_{\kq+1}^q= \operatorname*{arg\,max}_{n=1,\ldots,M}
\sum_{j=1}^L |\la \vd_n, \vr_{\jq}^{\kq}\ra|,
\ee
where $\vr_{\jq}^0= \vf_{\jq}$ and $\vr_{\jq}^{\kq}= 
\vf_{\jq}- \op_{\V_{\kq}^q} \vf_{\jq}$. 
Alternatively, the extension of OMP to multichannels 
which is known as  MMV-OMP \cite{CRE05,CH06}
  selects the index fulfilling
\be
\label{ompml}
\ell_{\kq+1}^q= \operatorname*{arg\,max}_{n=1,\ldots,M}
\sum_{j=1}^L |\la \vd_n, \vr_{\jq}^{\kq}\ra|^2.
\ee
The optimization of the OMP criterion to select  the 
atoms minimizing  the  norm of the residual error 
for each block goes with 
several names, according to the context were it was 
derived and the actual implementation. In one of 
the earliest
references \cite{CBL89} is called Orthogonal 
Least Square. In others
is called Order Recursive Matching Pursuit (ORMP) 
and in particular for multichannel signals MMV-ORMP
\cite{CRE05}. The 
implementation we adopt here is termed Optimized Orthogonal 
Matching Pursuit (OOMP) \cite{RNL02} and will be termed
OOMPMl for multichannel signals. 
 The approach selects the atoms  
$\ell_{\kq+1}^q$,
 for the approximation of block $q$, in other to 
minimize the norm of the residual error 
$\| \vF_q - \vF_q^{\kq+1}\|$ for that block. 
Those atoms correspond to the indices selected as:
\be
\label{oompml}
\ell_{\kq+1}^q=\operatorname*{arg\,max}_{\substack{n=1,\ldots,M\\ n\notin \Gamma_q}}
 \frac{\sum_{j=1}^L|\la \vd_n,\vr_{\jq}^{\kq}
\ra|^2}{1 - \sum_{i=1}^{\kq}
|\la \vd_n ,\vWt_i^q\ra|^2}
, \quad \qQ,
\ee
where $\Gamma_q=\{\ell_n^q\}_{n=1}^{\kq}$,  
$\vr_{\jq}^{\kq}= 
\vf_{\jq}- \op_{\V_{\kq}^q} \vf_{\jq}$, with  
$\vr_{\jq}^{0}= \vf_{\jq}$, and 
${\ds{\vWt_i^q= \frac{\vW_i^q}{\|\vW_i^q\|}}}$,
with $\vW_i^q,\, i=1,\ldots,\kq$ as in \eqref{GS}. The 
proof follows as in Theorem 2, but fixing the 
value of $q$ and taking the maximization over
the index. 

As will be discussed in the next section, 
the used of trigonometric 
dictionaries reduces the complexity of the 
calculations in \eqref{ompml1}, \eqref{ompml},
 and \eqref{oompml}.

The particularity of the OOMPMl implementation 
 being that
the coefficients of the atomic decomposition 
\eqref{atoq} are calculated 
using vectors \eqref{BW} which are adaptively
upgraded together with the selection of each new atom. 
For the $q$th-block the coefficients in the atomic 
decompositions \eqref{atoq} are computed as: 
\be
\label{coeop}
c_{\jq}^{\kq}(n)=\la\vB_{n}^{\kq,q}, \vf_{\jq} 
\ra,\,n=1,\ldots,\kq,
,\,\jL,
\ee
with $\vB_{n}^{\kq,q}$ as in \eqref{BW}. Thus, 
when the channels have similar sparsity structure 
by approximating all of them simultaneously 
the complexity is reduced.

The HBW-OMPMl/OOMPMl approach is implemented 
by the following steps: 
\begin{itemize}
\item[1)] Initialize the algorithm by 
selecting the `potential' first atom for 
 the atomic decomposition of every 
block $q$, according to criterion \eqref{ompml} or 
 \eqref{oompml}. For $\qQ$ set: $\kq=1,
\vW_1^q= \vB_1^{1,q}= \vd_{\ell_1^q}.$ 
\item[2)] Use criterion \eqref{hbwoompml} for 
selecting the block $\qs$ to 
upgrade the atomic decomposition by incorporating 
the atom corresponding to the index ${\ell^{\qs}_{k_{\qs}}}$.  If $k_{\qs}>1$ upgrade vectors \eqref{BW} for block $\qs$.
\item[3)]Increase $k_{\qs}\leftarrow k_{\qs}+1$ and 
select a new potential atom for the 
 atomic decomposition of block $\qs$, using the 
same criterion as in 1). 
Compute the corresponding 
$\vW_{k_{\qs}}^{\qs}$ (c.f. \eqref{GS}).
\item[4)] Check if, for a given 
$K$, the condition $\sum_{q=1}^Q \kq=K+1$ has been met. 
Otherwise repeat steps 2) - 4).
\item[5)] For each block, $\qQ$, 
 calculate the coefficients in \eqref{atoq} 
as in \eqref{coeop}.
\end{itemize} 
\subsection{Implementation details 
with Trigonometric Dictionaries}
In \cite{LRN15} we illustrate the clear advantage of 
approximating music using a mixed dictionary with 
components $\D^c$ and $\D^s$ as below
\begin{itemize}
\item
$\mathcal{D}^{c}=\{ \frac{1}{w^c(n)}\,
 \cos ({\frac{{\pi(2i-1)(n-1)}}{2M}}),\,i=1,\ldots,\Nq\}_{n=1}^{M}.$
\item
$
\mathcal{D}^{s}=\{\frac{1}{w^s(n)}\, \sin  ({\frac{{\pi(2i-1)n}}{2M}}), \,i=1,\ldots,\Nq\}_{n=1}^{M},$
\end{itemize}
where $w^c(n)$ and $w^s(n),\, n=1,\ldots,M$ are
normalization factors as given by
$$w^c(n)=
\begin{cases}
\sqrt{\Nq} & \mbox{if} \quad n=1,\\
\sqrt{\frac{\Nq}{2} + \frac{ 
\sin(\frac{\pi (n-1)}{M}) \sin(\frac{2\pi(n-1)\Nq}{M})}
{2(1 - \cos(\frac{ 2 \pi(n-1)}{M}))}} & \mbox{if} \quad n\neq 1.
\end{cases}
$$
$$w^s(n)=
\begin{cases}
\sqrt{\Nq} & \mbox{if} \quad n=1,\\
\sqrt{\frac{\Nq}{2} - \frac{ 
\sin(\frac{\pi n}{M}) \sin(\frac{2\pi n \Nq}{M})}
{2(1 - \cos(\frac{ 2 \pi n}{M}))}} & \mbox{if} \quad n\neq 1.
\end{cases}
$$
Fixing $M=2\Nq$ a dictionary redundancy four is constructed 
as $\D= \D^c \cup \D^s$. In addition to yielding highly 
sparse representation of melodic music,  
this trigonometric dictionary leaves room for 
reduction in the computational complexity 
of the algorithms 
and also in the storage demands. As discussed below, 
savings are made 
possible in an straightforward manner via the Fast Fourier 
Transform (FFT).

Given a vector $\vy\in \R^M$ we define
\be
\label{dft}
{\cal{F}}(\vy,n,M)= \sum_{j=1}^M y(j) e^{\im  2 \pi \frac{(n-1)(j-1)}{M}},\quad n=1,\ldots,M.
\ee
When $M=\Nq$ \eqref{dft} is the Discrete Fourier Transform 
of vector $\vy \in \R^{\Nq}$, which can be evaluated using
FFT. If $M>N_b$ we can still calculate \eqref{dft} via 
FFT by padding with $(M-N_b)$ zeros the vector  $\vy$. 
Accordingly, \eqref{dft} is a useful tool for calculating inner 
products with the atoms in dictionaries $\D^c$ and 
$\D^s$. For $n=1,\ldots,M$ it holds that
\be
\label{dct}
\sum_{j=1}^{\Nq}\cos{\frac{{\pi(2j-1)(n-1)}}{2M}}
y(j)= \Re \left(e^{-\im \frac{\pi (n-1)}{2M}}
{\cal{F}}(\vy,n,2M)\right)
\ee
and for $n=2,\ldots, M+1$
\be
\label{dst}
\sum_{j=1}^{\Nq}\sin{\frac{{\pi(2j-1)(n-1)}}{2M}}
y(j)= \Im \left(e^{-\im \frac{\pi (n-1)}{2M}}
{\cal{F}}(\vy,n,2M)\right),
\ee
where $\Re(z)$ indicates the real part of $z$ and
$\Im(z)$ its imaginary part.
The assistance of the FFT for performing 
the inner products \eqref{dct} and \eqref{dst}
reduces the complexity in calculating the 
maximizing function in \eqref{ompml},  
 (which is also the numerator  in \eqref{oompml})
  from $2M \Nq$ to  $2M (\log_2 2M + 1)$.  
Storing the sums 
$S_n^{\kq-1}=\sum_{i=1}^{\kq-1}|\la \vd_n,\vWt_{i}^q \ra|^2,\,n=1\,\ldots,M$ 
the denominator 
in the right hand side of 
\eqref{oompml} involves the calculation of 
$|\la \vd_n,\vWt_{\kq}^q \ra|^2,\,n=1\,\ldots,M$, 
which can also be computed via FFT. Hence the 
complexity for the calculation of 
the denominator in \eqref{oompml} is of
 the same order as that for the  calculation of  
the numerator. Criterion \eqref{oompml} in general 
yields higher sparsity, hence it reduces the cost 
in calculating vectors \eqref{BW}. 

The MATLAB function implementing 
the HBW-OOMPMl approach, named
HBW-OOMPMlTrgFFT when dedicated to
the above trigonometric dictionary,
has been made available on \cite{webpage}.
\section{A simple coding strategy}
\label{cod}
Previously to entropy encoding the coefficients resulting
from approximating a signal by partitioning, the
real numbers need to be converted into integers. This
operation is known as quantization. For the
numerical example of Sec.~\ref{NE} we adopt a simple
uniform quantization technique:
The absolute value coefficients
$|c_{\jq}(n)|,\,n=1\ldots,\kq,\,q=1,\ldots,Q,
\,j=1,\ldots,L$ are
converted to integers as follows:
\be
\label{uniq}
c^\Delta_{\jq}(n)= \lfloor \frac{|c_{\jq}(n)|}{\Delta} +\frac{1}{2} \rfloor,
\ee
where $\lfloor x \rfloor$ indicates the largest
integer number
smaller or equal to $x$ and  $\Delta$ is the quantization
parameter.
The signs of the coefficients, represented as
$\vs_{\jq},\,q=1,\ldots,Q,\,j=1,\ldots,L$,
are encoded separately using a binary alphabet.
As for the indices of the atoms, which are
common to the atomic
decompositions of all the channels,
they are firstly sorted in ascending
order
$\ell_{i}^q \rightarrow \til{\ell}_i^q,\,i=1,\ldots,\kq$,
which guarantees that, for each $q$ value,
$\til{\ell}_i^q < \til{\ell}_{i+1}^q,\,i=1,\ldots,\kq-1$.
This order of the indices induces
an order in the coefficients,
$\vc^\Delta_{\jq} \rightarrow \vct^\Delta_{\jq}$ and
in the corresponding signs $\vs_{\jq} \rightarrow
\vst_{\jq}$.  The advantage introduced by the
ascending order of the indices is that they can be
stored as smaller positive numbers by taking differences
between two consecutive values. Indeed,
by defining $\delta^q_i=\ellt^q_i- \ellt^q_{i-1},\,i=2,\ldots,\kq$ the follow string  stores the indices for block
$q$ with unique recovery
${\ellt^q_1, \delta^q_2, 
\ldots, \delta^q_{\kq}}$.
The number `0' is then used to separate the string
corresponding to different blocks and entropy code
 a long string, $\Sti$, which is built as
\be
\label{sti}
\Sti=[{\ellt_1}^1,\ldots,\delta^1_{k_1},0, \cdots,
0, \cdots,
{\ellt_1}^{k_Q},\ldots, \delta^Q_{k_Q}].
\ee
The corresponding quantized magnitude of the coefficients
of each
channel are concatenated in the $L$ strings  $\Stc,\, 
j=1,\ldots,L$ as follows:
\be
\label{stc}
\Stc
= [\ctq_{1,j}(1), \ldots, \ctq_{1,j}(k_1), \cdots,
\ctq_{{k_Q},j}(1), \ldots, \ctq_{{k_Q},j}(k_Q)].
\ee
Using `0' to store a positive sign and `1' to store
negative one, the signs  are placed
in the  $L$ strings, $\Sts, j=1,\ldots,L$  as
\be
\label{sts}
\Sts=[\st_{1,j}(1), \ldots, \st_{1,j}(k_1), \cdots,
\st_{{k_Q},j}(1), \ldots, \st_{{k_Q},j}(k_Q)].
\ee

The next encoding$/$decoding scheme
summarizes the above described procedure.\\\\
\newpage
{\bf{Encoding}}
\begin{itemize}
\item
Given a partition
$\vF_q \in \R^{\Nq \times L},\, q=1,\ldots,Q$
of a multichannel signal,
where for each $q$
the channels $\vf_{\jq} \in \R^{\Nq},\,j=1,\ldots,L$
are placed as columns of $\vF_q$,
approximate simultaneously all the channels through the 
 HBW-OOMPMlTrgFFT approach 
using $K=\sum_{q=1}^{Q}\kq$ atoms to obtain:
\be
\label{atc}
\vf_{\jq}^{\kq}=\sum_{n=1}^{\kq} c_{\jq}(n)\vd_{\ell_n},
\quad j=1,\ldots,L, \, q=1\ldots,Q.
\ee
\item
Quantize, as in \eqref{uniq},
the absolute vale coefficients in the above
equation to obtain
$c_{\jq}^{\Delta} (n),\,n=1,\ldots,\kq,\, j=1,\ldots,L,\, 
q=1,\ldots,Q$.
\item
For each $q$, sort the indices
$\ell_1^q,\ldots,\ell_{\kq}$ in ascending
oder to have a new order $\ellt_1^q,\ldots, \ellt_{\kq}$
and the re-ordered sets $\st_{\jq}(1), \ldots, \st_{\jq}(\kq)$,
and $\ct_{\jq}(1),\ldots,\ct_{\jq}(\kq)$, 
to create the strings:
$\Sti$, as in \eqref{sti}, and
$\Stc$, and $\Sts,\, j=1,\ldots,L$ as in \eqref{stc} and
\eqref{sts}, respectively. All these strings are
 encoded, separately, using arithmetic coding.
\end{itemize}
{\bf{Decoding}}
\begin{itemize}
\item
Reverse the arithmetic coding to recover strings
 $\Sti, \Stc, \Sts, j=1,\ldots,L$.
\item
Invert the quantization step as
$|\til{c}_{\jq}^{\mathrm{r}}(n)|= 
\Delta {\til{c}}_{\jq}^\Delta(n)$.
\item
Recover the partition of each channel
through the liner combination
$$\vf_{\jq}^{{\mathrm{r}},\kq} 
 = \sum_{n=1}^{\kq} \til{s}_{\jq}(n)|{\til{c}}_{\jq}^{\mathrm{r}}(n)| \vd_{\til{\ell}_{n}^q}.$$
\item
Assemble the recovered signal for each channel as
$$\vf_j^{\mathrm{r}}=\Join_{q=1}^Q \vf_{\jq}^{{\mathrm{r}},\kq},\,\jL$$
\end{itemize}
As already mentioned,  the quality of the 
recovered signal is assessed by the SNR measure,
which is calculated as
$$\text{SNR}=\log_{10} \frac{\| \vF\|_F^2}{\|\vF - \vFr\|_F^2}= 10 \log_{10}\frac{\sum_{j=1}^{L} 
\|\vf_{j}\|^2}{
\sum_{j=1}^{L} \|\vf_{j} 
-\vfr_{j}\|^2}.$$
\section{Numerical Example}
\label{NE}
This section is dedicated to illustrate the potential 
of the proposed codec for compressing melodic music 
with high quality recovery. The comparison with 
the state of the art is realized 
with respect to the OGG format \cite{xiph}.
 The reasons being: 
a) OGG is free licence. b) It is known to 
recover a signal of 
audible quality comparable to MP3 (superior for some 
opinions)
from a file of the same size. c)For a high quality setting
(e.g. more than 90$\%$) OGG produces a high SNR, 
which implies that the recovered signal 
is close to the original signal in the sense 
of the usual Euclidean distance.

The test clips, originally in WAV format and all sampled at 
44100 Hz, are listed in Table I. All the clips are 
stereo, with two channels. The SNR, in all the cases, is 
fixed by setting the OGG quality 90$\%$. For comparison 
purposes the 
Trigonometric Dictionary Codec (TDC) is tuned to 
reproduce the same SNR in each case.  
The partition unit is fixed as $\Nq=1024$ 
sample points.
The approximation routine 
is set to produce a SNR a few dBs higher 
than the required one
and the quantization parameter $\Delta$ is tuned to match 
the OGG's SNR. The file sizes are shown in Table I. 
The sizes corresponding to the TDC  are obtained using 
the Arith06 MATLAB function at the entropy 
coding step. The function is available on \cite{Karl}. 
Figure 1 shows the comparison bars 
in kbps (kilobit per sec) 
of the compression rate for the clips of Table I and 
the corresponding SNR values.
\begin{table}[!h]
\label{tab1}
\begin{center}
\begin{tabular}{|l|l||c|r|r|}
\hline
Clip&  SNR  & OGG & TDC \\ \hline \hline
C1  Electric Guitar &32.10dB& 212KB&72KB \\ \hline
C2  Harmonics Guitar&35.30dB&  359KB&126KB  \\ \hline
C3  Classic Guitar  &35.84dB& 723KB&89KB    \\ \hline
C4  Pop Piano Chord &32.53dB& 261KB&55KB    \\ \hline
C5  Cathedral Organ &34.38dB& 537KB&110KB   \\ \hline
C6  Orchestra Horns &39.07dB& 629KB&122KB   \\ \hline
C7  Ascending Jazz  &33.38dB& 102KB&28KB    \\ \hline
C8  Orchestrated    &35.56dB& 205KB&31KB    \\ \hline
C9  Classic Orchestra &34.63dB& 125KB&33KB    \\ \hline
C10 Trumpet Sax     &34.25dB&  78KB&32KB    \\ \hline
C11 Orchestra Entrance&34.21dB&74KB&19KB    \\ \hline
C12 Piazzola (Orches.) &31.90dB&205KB&58KB    \\ \hline
C13 Chopin (Piano)   &35.44dB& 414KB&58KB   \\ \hline \hline
\end{tabular}
\caption{Comparison of file sizes (in Kilobytes) compressing 
the clips to produce the same SNR. The SNR values arise 
by setting 90$\%$ quality for the OGG compression. Most of 
the clips are from {\tt{free-loops.com}}. C3 and C9 and C13 
are from sample WAV files on {\tt{onclassical.com}}}. 
\end{center}
\end{table}
\begin{figure}[!ht]
\begin{center}
\includegraphics[width=9cm]{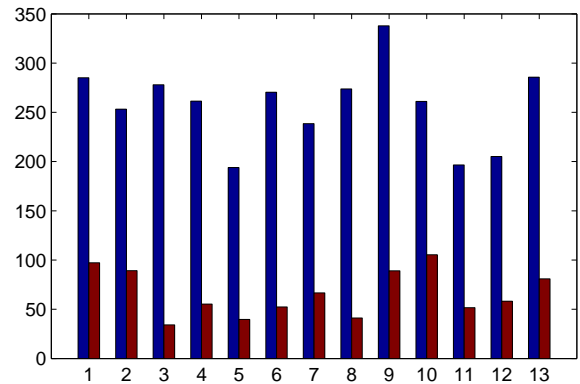}
\caption{Comparison bars, OGG vs TDC for 
the clips of Table 1. The vertical axis corresponds to 
the compression rate in kbps.}
\end{center}
\end{figure}
\section{Conclusions}
\label{con}
The proof of concept of the proposed TDC has been 
presented. Comparisons with the OGG Vorbis standard 
at 90$\%$ quality demonstrate the potential of the 
proposed codec: For all the short clips of melodic 
music that have been tested (in addition to those
in Table I) the TDC achieves remarkable reduction 
in the file size for the identical quality. 

Because this work focusses on assessing compression 
vs quality, the pursuit technique which has been 
applied at the approximation stage aims at producing 
high sparsity. At the entropy 
coding step compression performance was prioritized 
over speed. In addition to Arith06, the arithmetic encoder 
Arith07 and Huffman encoder Huff06, implemented by 
the MATLAB functions available on \cite{Karl}, 
have been tested. All these entropy coding 
techniques produce similar outputs.

{\bf{Note}}: The MATLAB functions for implementing
the TDC and running the numerical examples of 
Sec.~\ref{NE}, can be downloaded on \cite{webpage}. 

\subsection*{Acknowledgements}
The author is grateful to  
Karl Skretting for making available the
Arith06, Arith07, and Huff06 functions 
for entropy coding \cite{Karl}.
\bibliographystyle{IEEEbib}
\bibliography{revbib}
\end{document}